\documentclass[a4paper,onecolumn,10pt,pra,accepted=2024-03-04]{quantumarticle}
\pdfoutput=1

\usepackage[utf8]{inputenc}
\usepackage[numbers]{natbib}
\usepackage{amsfonts,amsmath,amsthm}
\newtheorem{lemma}{Lemma}
\usepackage{setspace}
\usepackage{braket}
\usepackage{dsfont}
\usepackage{graphicx}
\usepackage{svg}
\usepackage{float}
\usepackage{hyperref}
\hypersetup{
    colorlinks=true
}

\begin{document}

\title{Two-Particle Scattering on Non-Translation Invariant Line Lattices}
\author{Luna Lima e Silva}
\affiliation{Instituto de F\'isica, Universidade Federal Fluminense, Niter\'oi, RJ, 24210-340, Brazil}
\author{Daniel J. Brod}
\affiliation{Instituto de F\'isica, Universidade Federal Fluminense, Niter\'oi, RJ, 24210-340, Brazil}

\begin{abstract}
Quantum walks have been used to develop quantum algorithms since their inception, and can be seen as an alternative to the usual circuit  model; combining single-particle quantum walks on sparse graphs with two-particle scattering on a line lattice is sufficient to perform universal quantum computation. In this work we solve the problem of two-particle scattering on the line lattice for a family of interactions without translation invariance, recovering the Bose-Hubbard interaction as the limiting case. Due to its generality, our systematic approach lays the groundwork to solve the more general problem of multi-particle scattering on general graphs, which in turn can enable design of different or simpler quantum gates and gadgets. As a consequence of this work, we show that a CPHASE gate can be achieved with high fidelity when the interaction acts only on a small portion of the line graph.
\end{abstract}

\maketitle

\section{Introduction}

There are two well-known models of quantum walks, discrete- \cite{dtqw1,dtqw2} and continuous- \cite{qcdectrees2} time quantum walks. The former is closely related to classical random walks---it is composed of a quantum particle and a quantum coin, and discrete time unitary dynamics mimics a coin toss that determines the direction of a step taken by the particle. 
Continuous-time quantum walks, on the other hand, consist of a quantum particle that evolves by the Schrodinger equation, with a Hamiltonian on a lattice or graph that is analogous to the Laplacian in free space. This model of quantum walk is the focus of our work.

Quantum walks were first studied as a tool for building quantum algorithms. They were shown to exponentially outperform their classical counterparts \cite{qcdectrees1,qcdectrees2,qcdectrees3},  and an important milestone was the discovery that they can be used to perform universal quantum computation, first by Childs \cite{singleparticleqw}, later improved in collaboration with Gosset and Webb \cite{multiparticleqw}. 

In \cite{multiparticleqw}, the quantum gates and other gadgets operate mostly in the regime where one particle is propagating far from any others, essentially performing a single-particle quantum walk on a complicated graph. Interactions are required only for two-qubit controlled-phase gates, where two particles scatter off of each other on the simplest graph, a very long line, via a translation-invariant Hamiltonian. This (approximate) translation invariance implies conservation of the momentum which, together with energy conservation, severely restricts the outcome of the scattering, as the momenta of the particles are individually preserved.

In this work, we solve the more general problem of two-particle quantum walks on the infinite line with an interaction that does not display translation symmetry. More concretely, we consider two particles that propagate on an infinite line graph, but only interact in a finite set of contiguous vertices via the Bose-Hubbard Hamiltonian \cite{twoparticlehubbard}. Inspired by scattering theory on continuous space, we employ the Lippmann-Schwinger formalism \cite{sakurai} to find the S-matrix for this scenario. Crucially, in contrast to previous works \cite{levinsonthm2,varbanov, multiparticleqw} that only compute the S-matrix for fixed particle momenta, our approach explicitly treats the S-matrix as an operator on the full Hilbert space, allowing us to describe momentum exchanges between the particles.

After solving the finite-region scattering problem, we also investigate how the solution converges to the previously-known results for the infinite line \cite{multiparticleqw}. We analytically prove this convergence in the asymptotic limit, and numerically investigate the rate of convergence. Surprisingly, our results suggest that a few dozen interaction vertices might already be sufficient to obtain a high-fidelity controlled-phase gate.

By not relying on translation symmetry, our results provide new tools to analyse the scattering of particles on complex graphs, adding to the toolbox of \cite{multiparticleqw}. In particular, we 
lay the groundwork to solve the more general problem of multi-particle scattering on general graphs, which might lead to better and more efficient quantum gates and gadgets for quantum computation. We also note that multi-particle scattering theory has broad applications, e.g.\ in condensed matter and particle physics, and its discrete-space counterpart has been much less explored, which suggests applications for our results beyond the field of quantum computing.

This work is organized as follows. We review some basic notions about continuous time quantum walks in Sec.\ \ref{sec:scattering}, and we solve the scattering problem in Sec.\ \ref{subsec:LipSch}. 
In Sec.\ \ref{subsec:asymp} we show analytically that in the limit of interaction on all sites we recover known results, and in Sec.\ \ref{sec:numerics} we show numerical results indicating the corresponding rate of convergence.

\section{Single-particle scattering} \label{sec:scattering}

Let us begin by briefly motivating the general setup for continuous-time quantum walks \cite{singleparticleqw,multiparticleqw,levinsonthm2}. Single-particle quantum walks on a one dimensional lattice can be seen as the discrete analogue of free-particle evolution on continuous space. The latter is described by the Laplacian (up to some constants):
\begin{equation}
	H^{\textit{continuous space}}=-\dfrac{\hbar^2}{2m}\nabla ^2 = -\dfrac{\hbar^2}{2m}\dfrac{\partial^2}{\partial x^2}.
\end{equation}
To discretize the one dimensional Laplacian, we can look to the natural definition of the derivative as a limit
\begin{equation}
	\frac{\partial^2 f}{\partial x^2}=\lim_{\epsilon\rightarrow 0}\frac{[f(x+\epsilon)-f(x)]-[f(x)-f(x-\epsilon)]}{\epsilon^2}=\lim_{\epsilon\rightarrow 0}\frac{f(x+\epsilon)+f(x-\epsilon)-2f(x)}{\epsilon^2}.
\end{equation}
If we consider a particle on a one dimensional lattice, with spacing $\epsilon>0$ and amplitudes $\braket{j|\psi}=\psi(j)$ defined at every lattice point $j\epsilon$, with $j\in\mathbb{Z}$, a natural choice for the Hamiltonian is
\begin{equation}
	\bra{j}H^{\textit{lattice}}\ket{\psi}=-\dfrac{\hbar^2}{2m\epsilon^2}(\psi(j+1)+\psi(j-1)-2\psi(j)) ,
\end{equation}
or
\begin{equation}
	H^{\textit{lattice}}=-\dfrac{\hbar^2}{2m\epsilon^2}\sum_{j\in\mathbb{Z}} \ket{j}\bra{j+1}+\ket{j+1}\bra{j}-2\ket{j}\bra{j}.
\end{equation}

We now adopt some simplifications: the Hamiltonian above has a term proportional to the identity which adds an unimportant phase factor to the evolved state, so we drop it. We choose units such that $\hbar^2/2m\epsilon^2=1$. Also, following a convention of previous literature, we choose the negative of this Hamiltonian. With these changes, the Hamiltonian now reads
\begin{equation}
H^{\textit{lattice}}=\sum_{j\in\mathbb{Z}} \ket{j}\bra{j+1}+\ket{j+1}\bra{j}.
\end{equation}
Notice that this is exactly the \textit{adjacency matrix} for the infinite line graph, i.e., a matrix whose $ij$ entry is $1$ if $\{i,j\}$ is an edge of the graph and $0$ otherwise. The same procedure works for higher dimensional lattices, which suggests that the adjacency matrix is the natural choice of Hamiltonian for quantum walks on graphs.

Finding the corresponding eigenstates requires solving the Schr\"{o}dinger equation
\begin{equation}
	H^{\textit{lattice}}\ket{\psi}=E\ket{\psi}.
\end{equation}
which, in the position basis, results in the following difference equation
\begin{equation}
\begin{aligned}
	\psi(j+1)+\psi(j-1)&=E\psi(j) &&  \forall  j\in\mathbb{Z}.
\end{aligned}
\end{equation}
Considering the ansatz $\psi(j)=z^j$, there is a solution if $z+z^{-1}=E$. Since the energy must be a real number, either $z\in\mathbb{R}$ or $z$ lies on the unit circle. The only states that can be delta-function normalizable are the ones with $|z|=1$, so we can choose $z=e^{-ip}$ for some $p\in [-\pi,\pi)$.

From now on, the eigenstate with momentum $p$ is denoted by $\ket{p}$. The energy of this state is $E=z+z^{-1}=2\cos{p}$ and its amplitudes are given by 
\begin{equation}
\braket{j|p}=\frac{1}{\sqrt{2\pi}}e^{-ijp},
\end{equation}
with normalization $\braket{k|p}=\delta(k-p)$. Due to our choice of phase convention, the velocity of a wavepacket with momentum centered on $p$ is $\frac{dE}{d(-p)}=2\sin{p}$, which is positive for $p\in (0,\pi)$ and negative for $p\in (-\pi,0)$.

\section{Two-particle scattering} \label{sec:twoparticle}

\subsection{Two-particle free states} \label{subsec:freestates}

The usual Hamiltonian describing two noninteracting particles evolving with respect to Hamiltonians, say, $H_A$ and $H_B$, is the Kronecker sum $H^{(2)}=H_A\otimes\mathds{1}+\mathds{1}\otimes H_B$. In this case, given eigenstates for both particles $\ket{\psi}_A$ and $\ket{\phi}_B$ with energies $E_{\psi}$ and $E_{\phi}$, respectively, the tensor product $\ket{\psi}_A\otimes\ket{\phi}_B$ is a solution of the two-particle Schrodinger equation, with energy $E=E_{\psi}+E_{\phi}$. In our case, for distinguishable particles on a lattice, we write $\ket{p_1}\ket{p_2}=\ket{p_1p_2^{\textit{dist}}}$ with amplitudes
\begin{equation}
\begin{aligned}
\braket{j_1 j_2|p_1 p_2^{\textit{dist}}}=\frac{1}{2\pi} e^{-i(p_1j_1+p_2j_2)}.
\end{aligned}
\end{equation}
For bosons and fermions we have, respectively\footnote{For the special case where $p_1=p_2$, the correct states are $\ket{p_1 p_2^{\textit{boson}}}=\ket{p_1 p_2^{\textit{dist}}}$ and $\ket{p_1 p_2^{\textit{fermion}}}=0$.},
\begin{equation}
\begin{aligned}
\braket{j_1 j_2|p_1 p_2^{\textit{boson}}}=\frac{1}{2\pi} \dfrac{e^{-i(p_1j_1+p_2j_2)}+e^{-i(p_2j_1+p_1j_2)}}{\sqrt{2}},\\
\braket{j_1 j_2|p_1 p_2^{\textit{fermion}}}=\frac{1}{2\pi} \dfrac{e^{-i(p_1j_1+p_2j_2)}-e^{-i(p_2j_1+p_1j_2)}}{\sqrt{2}}.
\end{aligned}
\end{equation}
To avoid double counting, we always consider that $p_1 < p_2$ for bosonic and fermionic states. From now on, the notation $\ket{p_1 p_2}$ will be used to represent any of the states above when the distinction is unimportant.
It is also convenient to define a constant $b$ such that
\begin{equation}
\begin{aligned}
\braket{j j|p_1 p_2}=\frac{b}{2\pi}e^{-i(p_1+p_2)j},
\end{aligned}
\end{equation}
where $b=1$ for distinguishable particles, $b=\sqrt{2}$ for bosons and $b=0$ for fermions.

\subsection{Two-particle interacting quantum walk} \label{subsec:LipSch}

We now consider the scattering of two particles propagating on an infinite line with a Bose-Hubbard type interaction on $2L+1$ contiguous vertices, namely,
\begin{equation}
\begin{aligned}H=H^{(2)}+V:=\sum_{w=1}^2 \sum_{j\in\mathbb{Z}} \left(\ket{j}\bra{j+1}_w+\ket{j+1}\bra{j}_w\right) + U\sum_{j=-L}^L\ket{jj}\bra{jj},
\end{aligned}
\end{equation}
where $\ket{a}\bra{b}_1\equiv\ket{a}\bra{b}\otimes\mathds{1}$, $\ket{a}\bra{b}_2\equiv\mathds{1}\otimes\ket{a}\bra{b}$, so that $H^{(2)}$ is the Hamiltonian for two free particles on the line and $V$ is the interaction, which is proportional to some constant $U$. In other words, the particles interact with each other when they are on the same site, but only for the central $2L+1$ sites. 

We are interested in computing the S-matrix, defined by \cite{weinberg}:
\begin{equation}
\begin{aligned}
S_{k_1 k_2; p_1 p_2}=\braket{k_1 k_2 ^-|p_1 p_2^+},
\end{aligned}
\end{equation}
where $\ket{p_1,p_2^\pm}$ are scattering states, i.e., eigenstates of the interacting Hamiltonian which converge, in the asymptotic past ($+$) or future ($-$) to the free-particle states  $\ket{p_1 p_2}$. In other words, $S_{k_1 k_2; p_1 p_2}$ is the probability amplitude of preparing the state $\ket{p_1 p_2}$ in the distant past and measuring the state $\ket{k_1 k_2}$ in the distant future. 
The scattering and free-particle states can be related by the Lippmann-Schwinger equation
$$\ket{{p_1 p_2}^\pm}=\ket{p_1 p_2} + (E-H^{(2)}\pm i\epsilon)^{-1}V\ket{{p_1 p_2}^\pm},$$
where we take the limit $\epsilon\rightarrow 0$ at the end of the computation.
After some manipulations, we can arrive at a simpler expression for the S-matrix \cite{taylor}:
\begin{equation}
\begin{aligned}
S_{k_1 k_2; p_1 p_2}=\delta(k_1-p_1)\delta(k_2-p_2)-2\pi i\delta(E_{k_1k_2}-E_{p_1p_2})\bra{k_1 k_2}V\ket{{p_1 p_2}^+},
\label{S}
\end{aligned}
\end{equation}
where $E_{q_1q_2}=2\cos{q_1}+2\cos{q_2}$.

From Eq.\ (\ref{S}), and due to the nature of the interaction, we only need to compute amplitudes when particles are on the same site, namely, terms of the type $\braket{j j|p_1 p_2}$. From the definition of $V$,
\begin{equation}
\begin{aligned}
\bra{k_1 k_2}V\ket{{p_1 p_2}^+}=\frac{bU}{2\pi}\sum_{l=-L}^L e^{i(k_1+k_2)l}\braket{ll|{p_1 p_2}^+}.
\label{Smatrix}
\end{aligned}
\end{equation}
An equation for $\braket{ll|{p_1 p_2}^+}$ can be found by taking the inner product of the Lippmann-Schwinger equation with $\bra{ll}$ for $-L\leq l\leq L$:
\begin{equation}
\begin{aligned}
\braket{ll|{p_1 p_2}^+}=\frac{b}{2\pi}&e^{-i(p_1+p_2)l}\\&+\frac{U}{4\pi^2}\sum_{m=-L}^L \braket{mm|{p_1 p_2}^+}\int_{-\pi}^\pi\int_{-\pi}^\pi\frac{e^{-i(k_1+k_2)(l-m)}}{E-2\cos{k_1}-2\cos{k_2}+ i\epsilon}dk_1 dk_2,
\label{lipp}
\end{aligned}
\end{equation}
where we used the resolution of the identity $\mathds{1}=\int_{-\pi}^\pi\int_{-\pi}^\pi \ket{k_1 k_2^{\textit{dist}}}\bra{k_1 k_2^{\textit{dist}}}dk_1 dk_2$. Let us define new variables $k_\pm=\frac{k_1\pm k_2}{2}$ and the following family of integrals (discussed in more detail in Appendix \ref{apx:A})
\begin{equation}
\begin{aligned}
J(E,n)&=\dfrac{1}{2\pi}\lim_{\epsilon\rightarrow 0}\int_{-\pi}^\pi\int_{-\pi}^\pi\frac{e^{-i(k_1+k_2)n}}{E-2\cos{k_1}-2\cos{k_2}+ i\epsilon}dk_1 dk_2\\
&=\dfrac{1}{2\pi}\lim_{\epsilon\rightarrow 0}\int_{-\pi}^\pi e^{-2ik_+n} \int_{-\pi}^\pi\frac{1}{E-4\cos{k_+}\cos{k_-}+ i\epsilon}dk_- dk_+.
\label{J_pm}
\end{aligned}
\end{equation}

Finally, to find the S-matrix we must solve the system of $2L+1$ linear equations \eqref{lipp}. Simplifying notation, we can write
\begin{equation}
\begin{aligned}
\sum_{m=-L}^L (2\pi\delta_{lm}-UJ(E,l-m))x_m=b(c_{p_1 p_2})_l\,,
\label{linsys}
\end{aligned}
\end{equation}
where $x_m=\braket{mm|{p_1 p_2}^+}$ and $(c_{p_1 p_2})_l:=e^{-i(p_1+p_2)l}$. Note that $-L\leq m,l\leq L$. The coefficients $t(n):=2\pi\delta_{n0}-UJ(E,n)$ form an $N\times N$ \textit{Toeplitz} matrix
\begin{equation}
\begin{aligned}
T_N=\begin{bmatrix}
t(0) & t(-1) & t(-2) &  \cdots & t(-(N-1))\\
t(1) & t(0) & t(-1) & \cdots & t(-(N-2))\\
t(2) & t(1) & t(0) & \cdots & t(-(N-3))\\
\vdots & \vdots &\vdots & \ddots & \vdots \\
t(N-1) & t(N-2) & t(N-3) & \cdots & t(0)
\end{bmatrix},
\end{aligned}
\end{equation}
such that Eq.\ \eqref{linsys} can be written as the matrix equation $T_N x=bc_{p_1 p_2}$. Looking back at Eq.\ \eqref{Smatrix}, the factor related to the interaction is
\begin{equation}
\begin{aligned}
\bra{k_1 k_2}V\ket{{p_1 p_2}^+}=\frac{bU}{2\pi}c_{k_1 k_2}^\dagger x = \frac{b^2 U}{2\pi}c_{k_1 k_2}^\dagger T_N^{-1} c_{p_1 p_2},
\label{Smatrix2}
\end{aligned}
\end{equation}and finally the S-matrix for finite $N$ is
\begin{equation}
\begin{aligned}
S^N_{k_1 k_2; p_1 p_2}=\delta(k_1-p_1)\delta(k_2-p_2)-i b^2 Uc_{k_1 k_2}^\dagger T_N^{-1} c_{p_1 p_2}\delta(E_{k_1k_2}-E_{p_1p_2}).
\label{finsmat}
\end{aligned}
\end{equation}

Equation (\ref{finsmat}) shows that, in general, the particles can exchange momentum as a result of scattering, for any finite region of interaction. We show in the next section that, in the limit $L\rightarrow\infty$, that is no longer the case and the momentum of each particle is individually conserved\footnote{For distinguishable particles, they actually switch their momenta, but they still cannot trade an arbitrary amount of momentum.}.

\subsection{Asymptotic behaviour for \texorpdfstring{$L\rightarrow\infty$}{}} \label{subsec:asymp}

So far, our results are exact for the Bose-Hubbard interaction restricted to a finite region. Let us now explore the limit $L\rightarrow \infty$. We begin by using the known fact \cite{toep-circ} that Toeplitz matrices can be asymptotically approximated by \textit{circulant} matrices. One possible circulant approximation is the matrix
\begin{equation}
\begin{aligned}
C_N=\begin{bmatrix}
c(0) & c(1) & c(2) & \cdots & c(N-1)\\
c(N-1) & c(0) & c(1) & \cdots & c(N-2)\\
c(N-2) & c(N-1) & c(0) & \cdots & c(N-2)\\
\vdots & \vdots & \vdots & \ddots & \vdots \\
c(1) & c(2) & c(3) & \cdots & c(0)
\end{bmatrix},
\end{aligned}
\end{equation}
with coefficients given by
\begin{equation}
\begin{aligned}
c(n)=\begin{cases}t(0)\, &n=0,\\
t(n)+t(n-N)\, &n=1,\cdots,N-1.\end{cases}
\end{aligned}
\end{equation}
Every circulant matrix is diagonalized by the unitary Fourier matrix \cite{circ}, given by $F_{lm}=\frac{1}{\sqrt{N}}w^{-lm}$, where $w=e^{\frac{2\pi i}{N}}$ (in our case, it is convenient to define both indices running from $-L$ to $L$). That is, we can write $C_N=F^\dagger diag(\lambda_N)F$, where the eigenvalues of $C_N$ are
$$(\lambda_N)_m=\sum_{n=0}^{N-1} c(n)w^{nm}.$$
Therefore, we can approximate Eq.\ \eqref{Smatrix2} by a much simpler expression:
\begin{equation}
\begin{aligned}
\bra{k_1 k_2}V\ket{{p_1 p_2}^+} = \frac{b^2 U}{2\pi}c_{k_1 k_2}^\dagger T_N^{-1} c_{p_1 p_2} \approx \frac{b^2 U}{2\pi}c_{k_1 k_2}^\dagger F^\dagger diag(\lambda_N)^{-1}F c_{p_1 p_2} \text{, as } N\rightarrow\infty.
\label{Smatrix3}
\end{aligned}
\end{equation}

Let now us obtain the asymptotic behavior of the $\lambda_l$. Noting that $J(E,-n)=J(E,n)$, we have,
\begin{equation}
\begin{aligned}
(\lambda_N)_l=2\pi-U\sum_{n=-2L}^{2L} J(E,n)w^{nl}.\end{aligned}
\end{equation}
Using the definition in Eq.\ \eqref{J_pm}, switching the sum and the integral and then taking the limit of large $N$, we can write
\begin{equation}
\begin{aligned}
(\lambda_N)_l&=2\pi-\frac{U}{2}\int_{-\pi}^\pi\frac{1}{E-4\cos{\frac{\pi l}{N}}\cos{k_-}+i\epsilon}dk_- -\frac{U}{2}\int_{-\pi}^\pi\frac{1}{E+4\cos{\frac{\pi l}{N}}\cos{k_-}+i\epsilon}dk_-,
\label{lambdadelta}
\end{aligned}
\end{equation}
where we used the fact that $\frac{1}{T}\sum_{n\in\mathbb{Z}}e^{2\pi i \frac{n}{T}}=\delta_T(x)=\sum_{n\in\mathbb{Z}}\delta(x-nT)$ is the Dirac comb, the periodic version of the Dirac delta, and we used the identity $2\delta_{T}(2x)=\delta_{T}(x)+\delta_{T}(x+\frac{T}{2})$, proved in Appendix \ref{apx:B}. The computation of these  integrals is left to Appendix \ref{apx:A}. The result, as $\epsilon\rightarrow 0$, is 
\begin{equation}
\begin{aligned}
(\lambda_N)_l=2\pi- \dfrac{2\pi U}{\sqrt{E^2-16\cos^2{\frac{\pi l}{N}}}},
\end{aligned}
\end{equation}
where the square root is defined to have positive imaginary part if the radicand is negative.

Going back to Eq.\ \eqref{Smatrix3}, we have
\begin{equation}
\begin{aligned}
\bra{k_1 k_2}V\ket{{p_1 p_2}^+}=\frac{b^2 U}{2\pi N}\sum_{l,m,n=-L}^L w^{l(n-m)}e^{-i(k_1+k_2)m}e^{i(p_1+p_2)n}(\lambda_N)_l^{-1}.
\end{aligned}
\end{equation}
Note that the variables $m$ and $n$ are uncoupled, so we can sum over them independently, in the limit of large $L$, to obtain
\begin{equation}
\begin{aligned}
\bra{k_1 k_2}V\ket{{p_1 p_2}^+}=\frac{2\pi b^2 U}{N}\sum_{l=-L}^L \delta_{2\pi}\left(k_1+k_2+\frac{2\pi l}{N}\right)\delta_{2\pi}\left(p_1+p_2+\frac{2\pi l}{N}\right)(\lambda_N)_l^{-1}.
\end{aligned}
\end{equation}
We now replace the sum by an integral with variable $\ell=\frac{2\pi}{N}l$, $d\ell=\frac{2\pi}{N}$, such that
\begin{equation}
\begin{aligned}
\bra{k_1 k_2}V\ket{{p_1 p_2}^+}&=b^2 U\delta_{2\pi}(k_1+k_2-p_1-p_2)\int_{-\pi}^{\pi} \dfrac{\delta_{2\pi}(p_1+p_2+\ell)}{2\pi- \frac{2\pi U}{\sqrt{E^2-16\cos^2{\frac{\ell}{2}}}}}d\ell\\
&=\dfrac{b^2 U}{2\pi-\frac{2\pi U}{\sqrt{E^2-16\cos^2{\frac{p_1+p_2}{2}}}}}\delta_{2\pi}(k_1+k_2-p_1-p_2).
\end{aligned}
\end{equation}
Since $E^2=16\cos^2{p_+}\cos^2{p_-}\leq 16\cos^2{p_+}$, where $p_\pm=\frac{p_1\pm p_2}{2}$, and recalling that this square root was defined to have positive imaginary part, we can write
\begin{equation}
\begin{aligned}
\bra{k_1 k_2}V\ket{{p_1 p_2}^+}=\dfrac{b^2 U}{2\pi+\frac{\pi i U}{|\sin{p_1}-\sin{p_2}|}}\delta_{2\pi}(k_1+k_2-p_1-p_2),
\end{aligned}
\end{equation}
where we have used the identity $\sin{p_1}-\sin{p_2}=2\cos{p_+}\sin{p_-}$.

To compute the S-matrix of Eq.\ \eqref{S}, we can now use the following identity, proven in Appendix \ref{apx:B},
\begin{equation}
\begin{aligned}
\delta(E_{k_1k_2}-E_{p_1p_2})\delta_{2\pi}(k_1+k_2-p_1-p_2)=\frac{1}{2|\sin k_1-\sin k_2|}(\delta_{2\pi}&(k_1 - p_1)\delta_{2\pi}(k_2-p_2)\\+&\delta_{2\pi}(k_1 - p_2)\delta_{2\pi}(k_2-p_1)).
\end{aligned}
\end{equation}
Finally we plug this into Eq.\ \eqref{S} to obtain
\begin{equation}
\begin{aligned}
S_{k_1 k_2; p_1 p_2}&=\delta(k_1-p_1)\delta(k_2-p_2)-2\pi i\delta(E_{k_1k_2}-E_{p_1p_2})\bra{k_1 k_2}V\ket{{p_1 p_2}^+}\\
&=\delta(k_1-p_1)\delta(k_2-p_2)\left(1-\dfrac{ 2i b^2 U}{4|\sin{p_1}-\sin{p_2}|+2 i U}\right)  \\
&\hspace{0.5cm}-\delta(k_1-p_2)\delta(k_2-p_1)\left(\dfrac{ 2i b^2 U}{4|\sin{p_1}-\sin{p_2}|+2 i U}\right).
\label{Sid}
\end{aligned}
\end{equation}
This expression shows us that, in the limit where we recover translation invariance, the particles cannot exchange any amount of momentum; they either conserve both momenta individually, or at most exchange them.

For distinguishable particles we have $b=1$, and the reflection and transmission coefficients are given by, respectively:
\begin{equation}
\begin{aligned}
R(p_1,p_2)&=\dfrac{-iU}{2|\sin{p_1}-\sin{p_2}|+iU},\\
T(p_1,p_2)&=\dfrac{2|\sin{p_1}-\sin{p_2}|}{2|\sin{p_1}-\sin{p_2}|+iU}.
\end{aligned}
\end{equation}
On the other hand, for bosons we have $b=\sqrt{2}$, and since we considered ordered momenta for the bosonic states, namely, $p_1<p_2$ and $k_1<k_2$, we are left with the following S-matrix
\begin{equation}
\begin{aligned}
S_{k_1 k_2; p_1 p_2}=\delta(k_1-p_1)\delta(k_2-p_2)\dfrac{2|\sin{p_1}-\sin{p_2}|- i  U}{2|\sin{p_1}-\sin{p_2}|+ i U}.
\label{ideal}
\end{aligned}
\end{equation}
This expression recovers the two-particle result of \cite{twoparticlehubbard} and in the Appendix B of \cite{multiparticleqw}, but without restricting the particles to be counterpropagating [i.e.\ $p_1\in(-\pi,0)$ and $p_2\in(0,\pi)$]. Convergence of operators in infinite dimensional Hilbert spaces is not straightforward, so the fact that we can recover the correct S-matrices by this limiting process is not entirely obvious and is quite useful, as the limiting problem might be much more difficult to solve.

In the next section we analyse two natural questions about this result; how to measure the convergence and what is the rate of convergence.

\section{Numerical analysis} \label{sec:numerics}

In this section, we numerically analyse the behaviour of the S-matrix for a finite number of interaction sites. We start from Eq.\ \eqref{finsmat}, up until which every computation was exact for finite $L$. {Note that throughout this section we do not use the circulant approximation used previously, in Sec. \ref{subsec:asymp}---the purpose of that approximation was only to recover known results in the $L\rightarrow\infty$ limit.}

We begin by noting that it is not trivial to measure how close the finite-$L$ S-matrix from Eq.\ \eqref{finsmat} is to the limiting case described in Eq.\ \eqref{ideal}. For example, let us consider narrow wavepackets in momentum space; in the limiting case, they scatter off of each other and acquire a phase factor depending on both momenta; on the other hand, for any finite region of interaction, these wavepackets do not acquire a phase after scattering, as we will see shortly. Therefore, there always exists some state for which the S-matrices give consistently different results, which means that the finite-case S-matrix does not converge to the asymptotic S-matrix in operator norm, even though we showed they converge element-wise. {We sidestep this difficulty by taking a more pragmatic approach---we measure the distance between output states generated by the S-matrices when acting on a particular subspace of input states, namely, those that would arise in quantum walk based computation.}

If we analyze carefully the construction of a universal quantum computer proposed in \cite{multiparticleqw}, we see that only three families of states might propagate in a long linear graph, corresponding to zero, one or two particles. Within the dual-rail encoding used in \cite{multiparticleqw}, for the purposes of implementing a CPHASE gate, the presence (absence) of a particle in the linear graph is equivalent to the logical state $\ket{1}$ (resp.\ $\ket{0}$). More concretely: the logical two-qubit state $\ket{00}$ is represented by the absence of any particles. The logical states $\ket{01}$ and $\ket{10}$ are represented by a single particle moving to the left or right, respectively (this choice is arbitrary), and, following \cite{multiparticleqw}, we chose them to be in a state of the form
\begin{equation}
\begin{aligned}
\ket{k^\sigma}:=\dfrac{1}{\sqrt{2\sigma}}\int_{k-\sigma}^{k+\sigma} \ket{k'}dk',
\end{aligned}
\end{equation}
for momentum $k$ and width $2\sigma$. Finally, we choose the logical state $\ket{11}$ as the (normalized) bosonic state
\begin{equation}
\begin{aligned}
\dfrac{\ket{k_1^\sigma}\ket{k_2^\sigma}+\ket{k_2^\sigma}\ket{k_1^\sigma}}{\sqrt{2+2|\braket{k_1^\sigma|k_2^\sigma}|^2}}.
\label{bosonwavepacket}
\end{aligned}
\end{equation}

For concreteness, in our numerical computation we follow the choice made in \cite{multiparticleqw} and fix one particle to have momentum centered around $k_1=\frac{\pi}{4}$, the other one centered around $k_2=-\frac{\pi}{2}$, and $U=2+\sqrt{2}$. Setting these choices in \eqref{ideal}, a scattering eigenstate $\ket{p_1 p_2 ^{\textit{boson}}}$ would accrue a phase of $-i$. This scattering eigenstate is recovered by setting $\sigma \rightarrow 0$ in the equation above, so we expect to recover a perfect phase gate with phase $-i$ in the limit of large $L$ and small $\sigma$.

As our main figure of merit, we compute the average gate fidelity between the finite-case S-matrix and a CPHASE gate with phase $\phi$, as in \cite{brod2016a}:
\begin{equation}
\begin{aligned}
F(\phi):=\int d\psi \bra{\psi}C_\phi^\dagger S_N \ket{\psi}\bra{\psi}S_N^\dagger C_\phi  \ket{\psi},
\label{fidphi}
\end{aligned}
\end{equation}
where $S_N$ is the S-matrix given in Eq.\ \eqref{S} for $N=2L+1$ interaction sites, the integration is carried over the two-qubit state space with the Haar measure, and
\begin{equation}
\begin{aligned}
C_{\phi} = \begin{bmatrix}
1 & 0 & 0 &  0\\
0 & 1 & 0 & 0\\
0 & 0 & 1 & 0\\
0 & 0 &0 & e^{i\phi}
\end{bmatrix}.
\end{aligned}
\end{equation}
Note that we could instead compute the fidelity between the finite-case S-matrix and its asymptotic limit directly. However, this would give us the same qualitative results as just computing the fidelity as described above and setting $\phi=-\frac{\pi}{2}$  (cf.\ Figure \ref{fig:fid}). Therefore, we focus our analysis on the fidelity described in Eq.\ \eqref{fidphi}, as it gives us more flexibility by allowing us to choose different values for the parameter $\phi$.
Finally, using the fact that the S-matrix acts as the identity whenever there are less than two particles present, we can simplify the average gate fidelity as
\begin{equation}
\begin{aligned}
F(\phi)=\dfrac{1}{10}(6+3Re(e^{-i\phi}\mathcal{F})+|\mathcal{F}|^2),
\end{aligned}
\end{equation}{where $\mathcal{F}=\bra{11}S_N\ket{11}$ is the only quantity left to calculate; it can be computed using the expression of the S-matrix \eqref{finsmat} together with the definition of the bosonic wavepacket \eqref{bosonwavepacket}. Note that it is a function of the wavepacket width $\sigma$.}

\begin{figure}[h]
    \centering
    \includegraphics[width=0.6\textwidth]{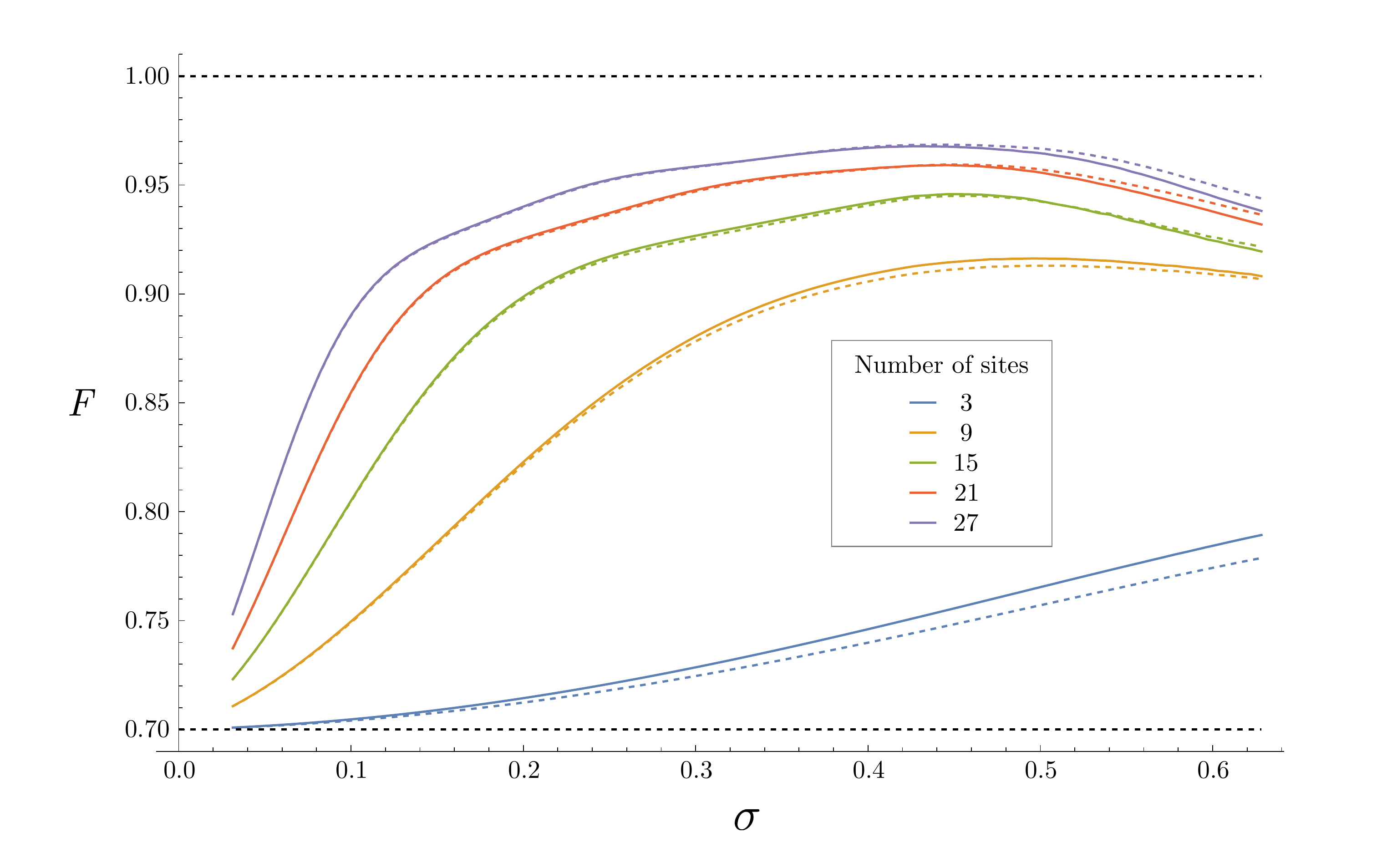}
    \caption{Solid line: Fidelity between $S_N$ and the CPHASE gate with the ideal phase $\phi=-\frac{\pi}{2}$ as a function of the wavepacket half width $\sigma$. Dashed line: Fidelity between $S_N$ and the limit S-matrix as a function of $\sigma$.}
    \label{fig:fid}
\end{figure}

Our first main result can be seen in Figure \ref{fig:fid}. There we observe that the fidelity for the phase $\phi=-\frac{\pi}{2}$, for small $\sigma$, increases for larger values of $L$, which tells us that in the limit $L\rightarrow\infty$, narrow momentum wavepackets scatter with the correct phase. We also verify this in Fig.\ \ref{fig:phase}, where we plot the highest fidelity relative to \emph{some} phase gate (left), with the corresponding phase $\phi$ (right). There, we see that, except for very small $\sigma$, the best phase approaches $\phi=-\frac{\pi}{2}$ for large $L$, as expected.

On the other hand, Fig.\ \ref{fig:fid} also tells us that wavepackets cannot be arbitrarily narrow, as the fidelity to a CPHASE gate with $\phi=-\frac{\pi}{2}$ tends to $7/10$ as $\sigma\rightarrow 0$. This can be shown analytically, since if we take the limit $\sigma\rightarrow 0$ for any finite $L$, the S-matrix acts as the identity on these wavepackets. We can interpret physically this as follows: since the interaction is proportional to the amplitude of both particles being at the same site, which vanishes for $\sigma\rightarrow 0$, any finite number of vertices of interaction is not enough to produce an effect asymptotically. This analysis is also corroborated by both graphs of Fig.\ \ref{fig:phase} which show that, indeed, whenever $\sigma$ is too small the closest phase gate has $\phi\rightarrow 0$, with fidelity close to one. 

\begin{figure}[h]
    \centering
    \includegraphics[width=\textwidth]{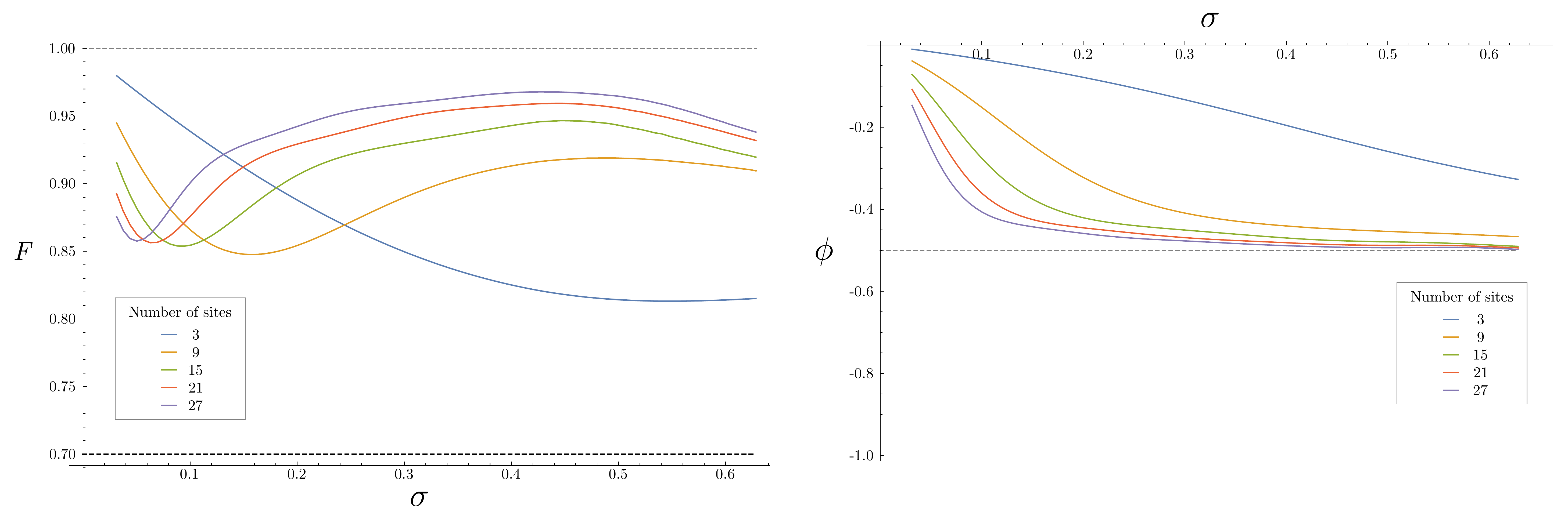}
    \caption{Left: Maximized fidelity over all phases as a function of $\sigma$. Right: Corresponding phase $\phi$ for which the fidelity is maximized.}
    \label{fig:phase}
\end{figure}

Two other important features stand out in the graphs in Fig. \ref{fig:fid} and Fig. \ref{fig:phase}. First, all of them display a plateau of behaviour above a certain wavepacket width, which suggests that a gate implemented this way would be robust to variations of the wavepacket preparation (contrast this with the results of \cite{brod2016a} for scattering of two \emph{photons} on a line of interaction sites, where the maximum only occurs for a single well-calibrated value of $\sigma$). Second, very high fidelities (over 95\%) are already achieved for relatively small number of interaction sites (under 30). This suggests that these two-qubit gates might be feasible with a modest amount of resources in a real-world implementation of this computational model. Furthermore, since the fidelity achieved is close to unity, it means that the scattering does not change appreciably the shape of the wavepackets, which is surprising given that the particles have many degrees of freedom that can become entangled.

Finally, Fig.\ \ref{fig:fid} suggests that the maximal fidelity is approaching $1$ as a function of $L$. In Fig.\ \ref{fig:cofid} we plot the infidelity, $1-F(-\frac{\pi}{2})$ (minimized over choice of $\sigma$) as a function of $L$ to verify that fact and extract a rate of convergence. Though we have few data points to extrapolate the results, we can infer that the fidelity is clearly converging to unity relatively fast.

\begin{figure}[h]
    \centering
    \includegraphics[width=0.7\textwidth]{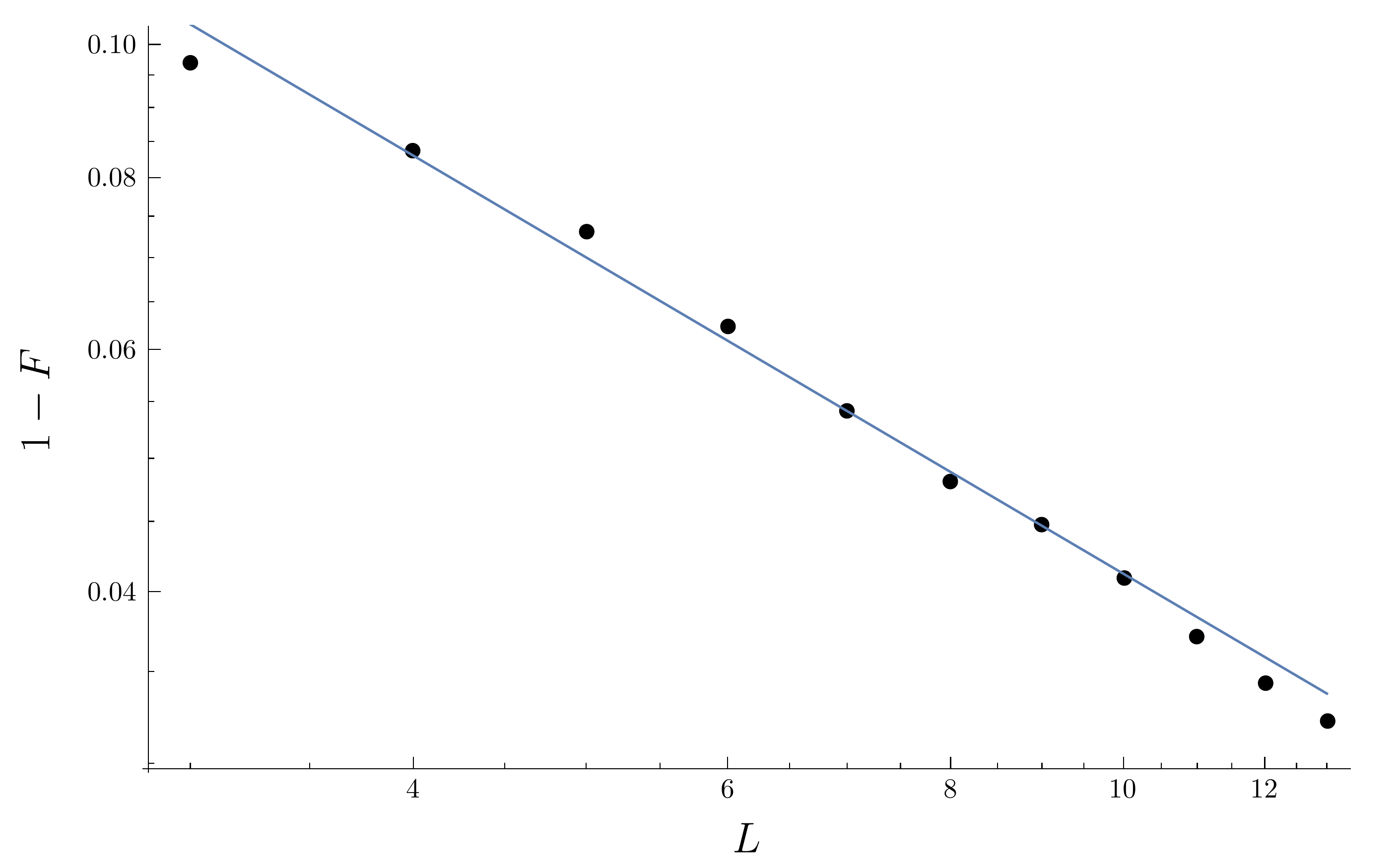}
    \caption{Loglog plot of infidelity as a funcion of $L$. The fitted curve corresponds to the power law $1-F = 0.24 L^{-0.76}$.}
    \label{fig:cofid}
\end{figure}

\section{Conclusion}

We have fully solved the scattering problem for two particles propagating on an infinite line interacting via the Bose-Hubbard Hamiltonian on a finite region of contiguous vertices. We analyzed the scattering in the asymptotic limit of infinite interaction vertices both analytically and numerically. The S-matrix we obtained reduces to previously-known results \cite{twoparticlehubbard,multiparticleqw} in the asymptotic limit (though the full S-matrix is not described explicitly in those papers, it can be inferred). By numerically analyzing the scattering in the finite region regime, we observed that a two-qubit gate with high fidelity can be already be achieved for very few interaction sites, and with wavepackets that may be easier to prepare than approximations to momentum eigenstates (e.g.\ with the aid of momentum filters \cite{singleparticleqw}).

Our work showcases how the Lippmann-Schwinger formalism can be used to solve the more general problem of multiparticle scattering on graphs, paving the way for a variety of future research directions. The more obvious next steps are to develop a full framework to treat multiparticle scattering on graphs with semi-infinite paths attached, such as those considered in \cite{multiparticleqw}. This might be used, for example, to develop new graphs that implement multi-qubit gates directly, or to simplify gadgets such as the momentum switches \cite{momswi}. A primitive in this direction would be to find nontrivial graphs with perfect transmission from input to output paths for two or more interacting particles.

Another goal is to filter out which ingredients make multiparticle scattering universal for quantum computing, developing new tools to study the computational complexity of quantum walks. For example, it is interesting to consider whether it is possible to perform universal computation where all particles have the same momenta (dispensing with the use of momentum switches), and for which values of momentum this holds. Conversely, one might ask whether there are values of momenta for which multiparticle scattering is classically simulable or otherwise nonuniversal \cite{bossamp}. 

Finally, we envision applications of these results beyond the field quantum computation. Scattering theory is a successful formalism used in other fields such as quantum optics, quantum field theory, condensed matter, and so on, and its discrete-space analogue is not so well-understood. In that sense, our result is analogous to that of \cite{brod2016a,brod2016b}, and it is an interesting question to try to extend this analogy further.

\begin{acknowledgments}
This work was supported by Brazilian funding agencies CNPq and FAPERJ.
\end{acknowledgments}
\bibliographystyle{quantum}
\bibliography{Scatteringrefs}

\appendix
\section{Computing \texorpdfstring{$J(E,n)$}{}} \label{apx:A}

Let us consider the family of integrals 
\begin{equation}
\begin{aligned}
J(E,n)=\dfrac{1}{2\pi}\lim_{\epsilon\rightarrow 0}\int_{-\pi}^\pi\int_{-\pi}^\pi\frac{e^{-i(k_1+k_2)n}}{E-2\cos{k_1}-2\cos{k_2}+ i\epsilon}dk_1 dk_2.
\label{J_pm2}
\end{aligned}
\end{equation}
It is easy to see that $J(E,n)=J(E,-n)=-\overline{J(-E,n)}$, so we only need to compute the integral for $n\ge 0$ and $E\ge 0$. Let us assume $E\ge 0$ in the following steps. Changing to the variables $k_\pm=\frac{k_1\pm k_2}{2}$, we have
\begin{equation}
\begin{aligned}
J(E,n)=\dfrac{1}{2\pi}\lim_{\epsilon\rightarrow 0}\int_{-\pi}^\pi e^{-2ik_+n} \int_{-\pi}^\pi\frac{1}{E-4\cos{k_+}\cos{k_-}+ i\epsilon}dk_- dk_+.
\label{J}
\end{aligned}
\end{equation}

Introducing a complex variable $z=e^{ik_-}$, the inner integral is written as
\begin{equation}
\begin{aligned}
I=i\oint \frac{1}{2\cos(k_+)z^2-(E+ i\epsilon)z+2\cos{k_+}}dz,
\label{I}
\end{aligned}
\end{equation}
where the contour is the unit circle in the positive orientation. The poles of the integrand are the roots of the denominator, $z_\pm=\frac{E+ i\epsilon}{4\cos{k_+}}\pm \sqrt{\left(\frac{E+i\epsilon}{4\cos{k_+}}\right)^2-1}$, where the square root is chosen to have positive real part. Let us assume for the moment that $\cos{k_+}>0$.

If $4\cos{k_+}<E$, then the roots are close to real and as $\epsilon\rightarrow 0$ we have $z_-<1<z_+$ as we can see in figure \ref{fig:roots}. So for sufficiently small $\epsilon$, the only root inside the unit circle is $z_-$.

If $4\cos{k_+}>E$, we can write the roots as $z_\pm=\frac{E+ i\epsilon}{4\cos{k_+}}\pm i\sqrt{1-\left(\frac{E+i\epsilon}{4\cos{k_+}}\right)^2}$, where the square root is chosen to have positive real part, similar to before. then the roots are not on the real line, and approach the unit circle as $\epsilon\rightarrow 0$, as pictured in figure \ref{fig:roots}. But for $\epsilon \neq 0$ and $|a|<1$, we have
\begin{equation}
\begin{aligned}
\left|a+i\epsilon \pm i\sqrt{1-(a+i\epsilon)^2}\right|^2=1\pm\frac{2\epsilon}{\sqrt{1-a^2}}+O(\epsilon^2),
\label{epsapprox}
\end{aligned}
\end{equation}
so, in both cases, the only root inside the contour is $z_-$. For the case $\cos{k_+}<0$, the roots $z_\pm$ switch roles, specifically, only the root $z_+$ is inside the contour.

\begin{figure}[h]
    \centering
    \includegraphics[width=\textwidth]{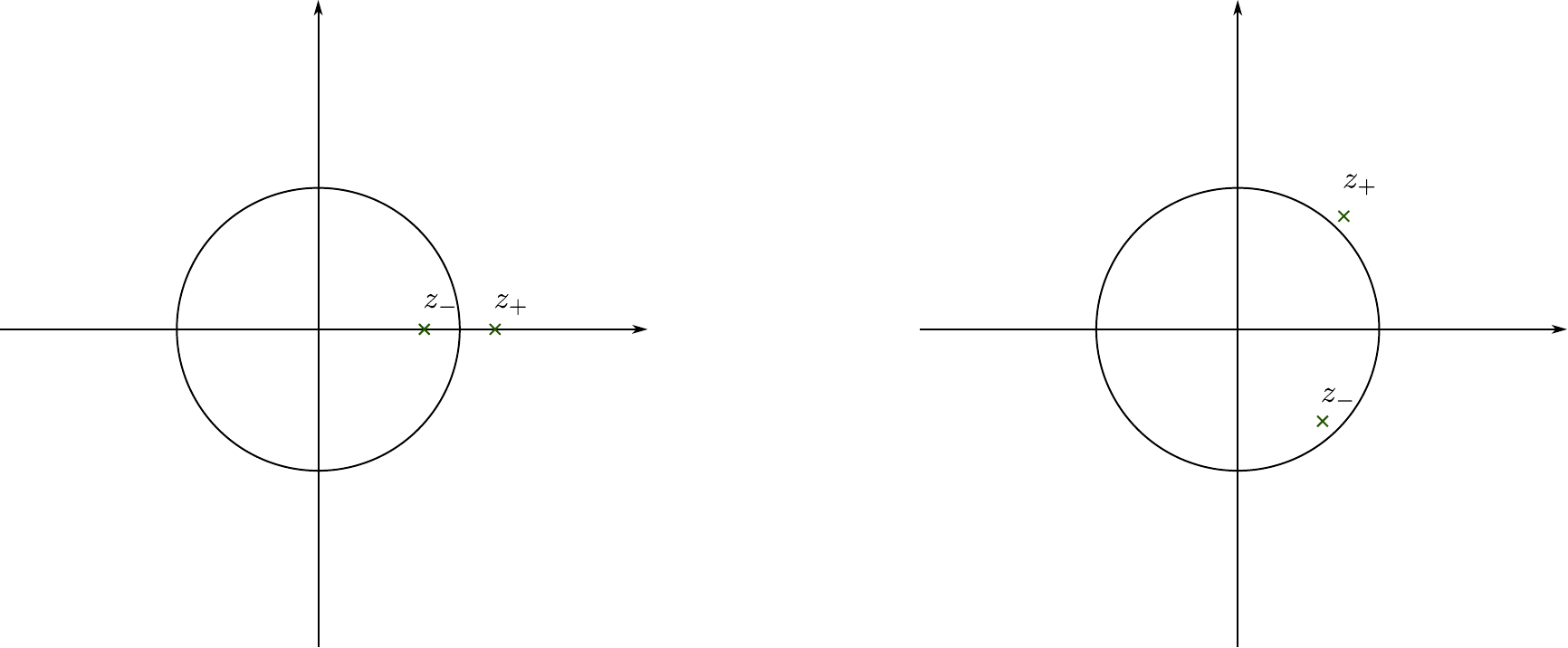}
    \caption{Left: Real roots $z_\pm$ for the case $0<4\cos{k_+}<E$ as $\epsilon\rightarrow 0$. Right: Complex roots $z_\pm$ for the case $4\cos{k_+}>E$.}
    \label{fig:roots}
\end{figure}

By residues, we have
\begin{equation}
\begin{aligned}
I&=2\pi i \begin{cases}\mathop{\mathrm{Res}}_{z = z_-}  \dfrac{i}{2\cos{k_+}(z-z_-)(z-z_+)}\, &\text{if} \, \cos{k_+}>0\\
\mathop{\mathrm{Res}}_{z = z_+}  \dfrac{i}{2\cos{k_+}(z-z_-)(z-z_+)}\, &\text{if} \, \cos{k_+}<0\end{cases}\\
&= \dfrac{2\pi}{\sqrt{\left(E+i\epsilon\right)^2-16\cos^2{k_+}}},
\end{aligned}
\end{equation}
where the square root is chosen as before; in the limit $\epsilon\rightarrow 0$, the square root is $\sqrt{E^2-16\cos^2{k_+}}$ for $4|\cos{k_+}|<E$ and $i\sqrt{16\cos^2{k_+}-E^2}$ for $4|\cos{k_+}|>E$. Now we can take the $\epsilon\rightarrow 0$ limit and we have
\begin{equation}
\begin{aligned}  
J(E,n)=\int_{0}^{\frac{\pi}{2}}\frac{\cos{2np}}{\sqrt{\left(\frac{E}{4}\right)^2-\cos^2{p}}}dp.
\end{aligned}    
\end{equation}

By decomposing $\cos{2np}=T_{2n}(\cos{p})$, where $T_m$ is the $m$-th Chebyshev polynomial of the first kind, we only need to compute the integrals
\begin{equation}
\begin{aligned}
\int_{0}^{\frac{\pi}{2}}\frac{\cos^{2n}{p}}{\sqrt{\left(\frac{E}{4}\right)^2-\cos^2{p}}}dp=\frac{1}{i\sqrt{1-\left(\frac{E}{4}\right)^2}}\left(\int_{0}^{\frac{\pi}{2}}\frac{\cos^{2n}{p}}{\sqrt{1-k^2\sin^2{p}}}dp\right)^*=:\frac{1}{i\sqrt{1-\left(\frac{E}{4}\right)^2}}(C_{n})^*,      
\end{aligned}
\end{equation}
where $k=\sqrt{\frac{16}{16-E^2}}$ is called the \textit{elliptic modulus}, and the square root of a negative number is taken to have positive imaginary part at all steps. The integrals $C_{n}$ are related to the complete elliptic integrals of the first and the second kind, $K(k)$ and $E(k)$ respectively. They are defined as
\begin{equation}
\begin{aligned}
K(k)&=\int_0^{\frac{\pi}{2}}\frac{dp}{\sqrt{1-k^2\sin^2{p}}},\\
E(k)&=\int_0^{\frac{\pi}{2}}\sqrt{1-k^2\sin^2{p}}\,dp.
\end{aligned}
\end{equation}

The sequence of integrals $C_n$ can be computed with the following recurrence relation (see equation 312.05 of \cite{byrd}):
\begin{equation}
\begin{aligned}
C_0&=K(k),\\
C_1&=\frac{1}{k^2}\left(E(k)-(1-k^2)K(k)\right),\\
C_n&= \frac{(2n-2)(2k^2-1)C_{n-1}-(2n-3)(1-k^2)C_{n-2}}{(2n-1)k^2}\qquad\forall n\ge 2.
\end{aligned}
\end{equation}
Finally, the integral $J(E,n)$ for $E\ge 0$ can be computed as
\begin{equation}
\begin{aligned}
J(E,n)=\frac{4^n n}{i\sqrt{1-\left(\frac{E}{4}\right)^2}}\sum_{m=0}^n\frac{(-1)^m(2n-m-1)!}{4^m m! (2n-2m)!}(C_{n-m})^*.    
\end{aligned}    
\end{equation}

\section{Some Dirac comb identities}\label{apx:B}
The Dirac comb is just a periodic extension of the Dirac delta, and one possible representation is given by its Fourier series:
\begin{equation}
\begin{aligned}
\delta_{T}(x)=\frac{1}{T}\sum_{n\in\mathbb{Z}}e^{i 2\pi n\frac{x}{T}},
\label{1id}
\end{aligned}
\end{equation}
where $T$ is the period of the Dirac comb.

One intuitive identity we used (in \eqref{lambdadelta}) is the scaling of the Dirac comb by a factor of two, $2\delta_{T}(2x)=\delta_{\frac{T}{2}}(x)=\delta_{T}(x)+\delta_{T}(x+\frac{T}{2})$, similar to the related scaling identity for the Dirac delta, but also changing the frequency. This can be easily proved using the Fourier series representation:
\begin{equation}
\begin{aligned}
\delta_{T}(x)+\delta_{T}(x+\frac{T}{2})&=\frac{1}{T}\sum_{n\in\mathbb{Z}}e^{i 2\pi n\frac{x}{T}} +\frac{1}{T}\sum_{n\in\mathbb{Z}}e^{i 2\pi n\frac{x+\frac{T}{2}}{T}} \\
&=\frac{1}{T}\sum_{n\in\mathbb{Z}}e^{i 2\pi n\frac{x}{T}} +\frac{1}{T}\sum_{n\in\mathbb{Z}}(-1)^n e^{i 2\pi n\frac{x}{T}} \\
&=\frac{2}{T}\sum_{n\, even}e^{i 2\pi n\frac{x}{T}} \\
&=\frac{2}{T}\sum_{n\in\mathbb{Z}}e^{i 4\pi n\frac{x}{T}} =2\delta_{T}(2x).
\end{aligned}
\end{equation}

Another identity was used (in \eqref{Sid}) :
\begin{equation}
\begin{aligned}
\delta(E_{k_1k_2}-E_{p_1p_2})\delta_{2\pi}(k_1+k_2-p_1-p_2)=\frac{1}{2|\sin k_1-\sin k_2|}(\delta_{2\pi}&(k_1 - p_1)\delta_{2\pi}(k_2-p_2)\\+&\delta_{2\pi}(k_1 - p_2)\delta_{2\pi}(k_2-p_1)).
\end{aligned}
\end{equation}
where $E_{q_1q_2}=2\cos{q_1}+2\cos{q_2}=4\cos{q_+}\cos{q_-}$ with $q_\pm=\frac{q_1\pm q_2}{2}$.

Starting with the left side, we can simplify it
\begin{equation}
\begin{aligned}
&\delta(E_{k_1k_2}-E_{p_1p_2})\delta_{2\pi}(k_1+k_2-p_1-p_2)\\&\hspace{1cm}=\delta(4\cos{k_+}\cos{k_-}-4\cos{p_+}\cos{p_-})\delta_{2\pi}(2k_+-2p_+)\\
&\hspace{1cm}=\delta(4\cos{k_+}\cos{k_-}-4\cos{p_+}\cos{p_-})\frac{1}{2}(\delta_{2\pi}(k_+-p_+)+\delta_{2\pi}(k_+-p_++\pi))\\
&\hspace{1cm}=\frac{1}{8|\cos{k_+}|}(\delta(\cos{k_-}-\cos{p_-})\delta_{2\pi}(k_+-p_+)\\&\hspace{3cm}+\delta(\cos{k_-}+\cos{p_-})\delta_{2\pi}(k_+-p_++\pi)),
\end{aligned}
\end{equation}

where we used the identity \eqref{1id} on the second line.

The first delta of cosines can be further simplified to a sum of Dirac combs:
\begin{equation}
\begin{aligned}
\delta(\cos{k_-}-\cos{p_-})=\frac{1}{|\sin{k_-}|}(\delta_{2\pi}(k_--p_-)+\delta_{2\pi}(k_-+p_-)),
\end{aligned}
\end{equation}
and similarly for the second product. So we have the full expression
\begin{equation}
\begin{aligned}
\frac{1}{8|\cos{k_+}\sin{k_-}|}(&\delta_{2\pi}(k_--p_-)\delta_{2\pi}(k_+-p_+)+
\delta_{2\pi}(k_-+p_-)\delta_{2\pi}(k_+-p_+)+\\
&\delta_{2\pi}(k_--p_-+\pi)\delta_{2\pi}(k_+-p_++\pi)+
\delta_{2\pi}(k_-+p_-+\pi)\delta_{2\pi}(k_+-p_++\pi)),
\label{deltaexp}
\end{aligned}
\end{equation}
and we want to go back to the original variables, $k_1,k_2,p_1$ and $p_2$.

We compute in detail the first product of the above expression.

\begin{lemma} \label{lem}
\begin{equation}
\begin{aligned}
\frac{1}{2}\delta_{2\pi}(k_--p_-)\delta_{2\pi}(k_+-p_+)=&\delta_{4\pi}(k_1-p_1)\delta_{4\pi}(k_2-p_2)\\&+\delta_{4\pi}(k_1-p_1+2\pi)\delta_{4\pi}(k_2-p_2+2\pi).    
\end{aligned}    
\end{equation}

\end{lemma}

\begin{proof}
Let us expand the factors on the right hand side as Fourier series:

\begin{equation}
\begin{aligned}
\frac{1}{(4\pi)^2}\sum_{n,m\in\mathbb{Z}}e^{i n\frac{k_1-p_1}{2}}&e^{i m\frac{k_2-p_2}{2}}(1+(-1)^{n+m}) =\frac{1}{(4\pi)^2}\sum_{n,l\in\mathbb{Z}}e^{i n\frac{k_1-p_1}{2}}e^{i (l-n)\frac{k_2-p_2}{2}}(1+(-1)^{l}) \\
&=\frac{2}{(4\pi)^2}\sum_{\substack{n\in\mathbb{Z}\\ l\,even}} e^{i n (k_--p_-)}e^{i l\frac{k_+-k_--p_++p_-}{2}}\\
&=\frac{2}{(4\pi)^2}\sum_{n,l\in\mathbb{Z}}e^{i n (k_--p_-)}e^{i l (k_+-k_--p_++p_-)}\\
&=\frac{1}{2}\frac{1}{(2\pi)^2}\sum_{q,l\in\mathbb{Z}}e^{i q (k_--p_-)}e^{i l (k_+-p_+)}=\frac{1}{2}\delta_{2\pi}(k_--p_-)\delta_{2\pi}(k_+-p_+),
\end{aligned}
\end{equation}
where we set $l=m+n$ in the first line and $q=n-l$ in the last line.
\end{proof}

Similarly, by substituting $k_1 \rightarrow k_1+2\pi$ in Lemma \ref{lem}, we get an identity for the third delta product of Eq.\ \eqref{deltaexp}:
\begin{equation}
\begin{aligned}
\frac{1}{2}\delta_{2\pi}(k_--p_-+\pi)\delta_{2\pi}(k_+-p_++\pi)=&\delta_{4\pi}(k_1-p_1+2\pi)\delta_{4\pi}(k_2-p_2)\\&+\delta_{4\pi}(k_1-p_1)\delta_{4\pi}(k_2-p_2+2\pi).
\end{aligned}    
\end{equation}

So, summing the first and third delta products in Eq.\ \eqref{deltaexp}, we get a $2\pi$ periodic expression instead of a $4\pi$ periodic one:
\begin{equation}
\begin{aligned}
\frac{1}{2}\delta_{2\pi}(k_--p_-+\pi)\delta_{2\pi}(k_+-p_++\pi)+\frac{1}{2}\delta_{2\pi}(k_--p_-)&\delta_{2\pi}(k_+-p_+)\\&=\delta_{2\pi}(k_1-p_1)\delta_{2\pi}(k_2-p_2).
\end{aligned}
\end{equation}

Similarly, the same steps can be done for the second and fourth delta products in Eq.\ \eqref{deltaexp}, by exchanging $k_1$ and $k_2$. Finally, we can simplify Eq.\ \eqref{deltaexp} to

\begin{equation}
\begin{aligned}
\frac{1}{8|\cos{k_+}\sin{k_-}|}(2\delta_{2\pi}(k_1 - p_1)\delta_{2\pi}(k_2-p_2)+2\delta_{2\pi}(k_1 - p_2)\delta_{2\pi}(k_2-p_1))\\
=\frac{1}{2|\sin k_1-\sin k_2|}(\delta_{2\pi}(k_1 - p_1)\delta_{2\pi}(k_2-p_2)+\delta_{2\pi}(k_1 - p_2)\delta_{2\pi}(k_2-p_1)).
\end{aligned}
\end{equation}

\end{document}